\documentclass{ebarticle}
\usepackage{ebutf8,crmac}

\title{Functions as proofs as processes}
\author{Emmanuel Beffara}
\institute{%
  Institut de Math\'{e}matiques de Luminy \\
  UMR6206, Universit\'{e} Aix-Marseille II \& CNRS}
\date{January 5, 2007}

\newcommand\LET[2]{{\operatorname{let}\,#1\mathrel{\operatorname{in}}#2}}
\newcommand\CASE[2]{%
  {\operatorname{case}\,#1\mathrel{\operatorname{of}}\{#2\}}}
\DeclareMathOperator\inj{inj}

\begin{document}
\maketitle

\begin{abstract}
  This paper presents a logical approach to the translation of functional
  calculi into concurrent process calculi.
  The starting point is a type system for the \pii-calculus closely related to
  linear logic.
  Decompositions of intuitionistic and classical logics into this system
  provide type-preserving translations of the \lam- and \lm-calculus, both for
  call-by-name and call-by-value evaluation strategies.
  Previously known encodings of the \lam-calculus are shown to correspond to
  particular cases of this logical embedding.
  The realisability interpretation of types in the \pii-calculus provides
  systematic soundness arguments for these translations and allows for the
  definition of type-safe extensions of functional calculi.
\end{abstract}

\section{Introduction} 

The \pii-calculus was introduced in the late 1980's as a core model of concurrent
computation, in the same way as the \lam-calculus is a core model of functional
computation.
Soon afterwards, Milner showed in the seminal paper
\emph{Functions as processes}~\cite{mil90:fun} that the \lam-calculus could be
precisely encoded into the \pii-calculus.
Around the same time, Girard introduced linear logic as a logic to study fine
properties of denotational models of intuitionistic logic.
Indeed the ideas behind it led to significant insights on the
structure and semantics of the \lam-calculus and functional computation, along
the lines of the \emph{functions as proofs} slogan.
It might not be a coincidence that these two contributions
appeared at the same time, and intuitions from one have been seen
in the other from the beginning.
Formal connections appeared some years later, giving formal ground to the
\emph{proofs as processes} idea, in particular in work by
Abramsky~\cite{abr93:cill,abr94:pap} and in a notable contribution by Bellin
and Scott as an encoding of proof nets in the \pii-calculus~\cite{bs94:pi-ll}.

The purpose of this paper is to present a formal way to make these pieces fit
together.
We use a recent version of the proofs-as-processes
approach~\cite{bef05:cmll,bef05:lrc} as a way to make a link between a form
of \pii-calculus (with more symmetry and expressiveness)
and a form of linear logic (with the significant difference that formulas have
arities).
In this framework, we adapt previous work by Danos, Joinet and Schellinx on
the translation of classical logic into linear
logic~\cite{djs95:lktq,djs96:decons}.
We show that, when considering each logic as a type system, we can extract a
family of typed translations of the \lam- and \lm-calculi into the \pii-calculus.
The now familiar duality~\cite{ch00:duality} between call-by-name and
call-by-value appears clearly in our system, moreover several previously
known translations are shown to fit in as particular instances of the general
technique.

Our type system for the \pii-calculus was developed by realisability as a logic
of behaviours of concurrent processes.
We show that this realisability construction can be used to prove properties
of the considered execution models of the \lam-calculus.
We also argue that realisability provides a way to introduce new constructs in
functional calculi while keeping the type system semantically correct.

\section{Framework} 

\subsection{The calculus} 

The concurrent calculus we use, hereafter named \pieq-calculus, is a
formulation of \pii-calculus with explicit fusions (à la Gardner and
Wischik~\cite{gw00:explicit}) with binding input and output.
We assume an infinite set $\Names$ of names, ranged over by the letters
$u,v,x,y,z$.
The calculus is generated by the following grammar:
  \begin{align*}
    &\text{actions:}& α ::={}
    & u(x_1…x_n) &&\text{input} \\
    &&& \bar u(x_1…x_n) &&\text{binding output} \\
    &\text{processes:}& p,q ::= {}
    & α.p,\quad !α.p &&\text{linear action, guarded replication} \\
    &&& 1, \quad p|q, \quad \new{x}p
      &&\text{inaction, parallel composition, hiding} \\
    &&& x\fuse y &&\text{name unification}
  \end{align*}

\begin{table}[t]
  Parallel composition and scoping:
  \begin{align*}
    p|q &≡ q|p &
    (p|q)|r &≡ p|(q|r) &
    p|1 &≡ p \\
    \new{x}\new{y}p &≡ \new{y}\new{x}p &
    \new{x}(p|q) &≡ p|\new{x}q \quad\text{if } x∉\fv(p) &
    \new{x}1 &≡ 1
  \end{align*}
  Equators:
  \begin{align*}
    1 &≡ x\fuse x &
    x\fuse y &≡ y\fuse x &
    x\fuse y | p[x/z] &≡ x\fuse y | p[y/z]
  \end{align*}
  Replication and reduction:
  \begin{align*}
    !α.p &≡ α.(p|!α.p) &
    u(\vec x).p | \bar u(\vec x).q &→ \new{\vec x}(p|q)
  \end{align*}
  \caption{Structural congruence and reduction.}
  \label{table:congr}
\end{table}

The operational semantics of the calculus is defined as a reduction relation
up to structural congruence, with the rules in table~\ref{table:congr}.
The reduction relation is the smallest relation $→$ that is closed under
structural congruence, parallel composition and hiding and that contains
$ u(\vec x).p | \bar u(\vec x).q → \new{\vec x}(p|q) $.
We consider a strong bisimilarity relation $≅$ whose precise definition (that
can be found in the appendix) uses a labelled transition system.
The point is that $≅$ is a congruent equivalence such that $p≅q$ implies that
for each reduction $p→p'$ there is a reduction $q→q'$ with $p'≅q'$.

We use this calculus instead of a more standard form of \pii-calculus because it
provides a clear distinction between synchronisation and name substitution.
It also allows a cleaner type system.
By combining binding actions and equators, we get usual non-binding actions
with their usual semantics, by defining
\[
  \bar u⟨x_1…x_n⟩ := \bar u(y_1…y_n).(x_1\fuse y_1|…|x_n\fuse y_n)
\]

\subsection{The type system: linear logic with arities} 

We assume a set $\Vars$ of type variables, ranged over by $X$ or $Y$.
The language of formulas is generated by the following grammar:
\[
  A,B ::= X \mid X^⊥ \mid A⊗B \mid A⅋B \mid ↓A \mid ↑A \mid !A \mid ?A
  \mid ∃X.A \mid ∀X.A
\]
Each variable is supposed to have a fixed arity.
Given an arity function $\ar:\Vars→\NN$, the arity of a formula is
defined as
\begin{gather*}
  \ar(\dagger A) := 1 \qquad\text{with } \dagger∈\{↑,↓,?,!\} \\
  \ar(∃X.A) := \ar(∀X.A) := \ar(A) \\
  \ar(A⊗B) := \ar(A⅋B) := \ar(A) + \ar(B)
\end{gather*}
The dual (or linear negation) is the involution $(⋅)^⊥$ defined
as $X^{⊥⊥}:=X$ and
\begin{align*}
  (A ⊗ B)^⊥ &:= A^⊥ ⅋ B^⊥ &
  (↓A)^⊥ &:= ↑(A^⊥) \\
  (∀X.A)^⊥ &:= ∃X.(A^⊥) &
  (!A)^⊥ &:= ?(A^⊥)
\end{align*}

A type $Γ$ is a sequence $\vec x_1:A_1,…,\vec x_n:A_n$ where each $A_i$ is a
formula and each $\vec x_i$ is a sequence of names of length $\ar(A_i)$.
All the names occurring in all the $\vec x_i$ must be distinct.
$↑Γ$ denotes a sequent where all formulas have the form $↑A$ or $?A$, and
$?Γ$ denotes a sequent where all formulas have the form $?A$.
A typing judgement is written $p⊢Γ$, where $p$ is a process and $Γ$ is a
type.
A process $p$ has type $Γ$ if $p⊢Γ$ is derivable by the rules of
table~\ref{table:typing}.
\begin{table}[p]
  Axiom and cut:
  \[
    \begin{prooftree}
      \Infer0{ u_1\fuse v_1|…|u_k\fuse v_k ⊢ \vec u:X^⊥, \vec v:X }
    \end{prooftree}
  \qquad
    \begin{prooftree}
      \Hypo{ p ⊢ Γ, \vec x:A }
      \Hypo{ q ⊢ \vec x:A^⊥, Δ }
      \Infer2{ \new{\vec x}(p|q) ⊢ Γ, Δ }
    \end{prooftree}
  \]
  Multiplicatives:
  \[
    \begin{prooftree}
      \Hypo{ p ⊢ Γ, \vec x:A }
      \Hypo{ q ⊢ Δ, \vec y:B }
      \Infer2{ p|q ⊢ Γ, Δ, \vec x\vec y:A⊗B }
    \end{prooftree}
  \qquad
    \begin{prooftree}
      \Hypo{ p ⊢ Γ, \vec x:A, \vec y:B }
      \Infer1{ p ⊢ Γ, \vec x\vec y:A⅋B }
    \end{prooftree}
  \]
  Actions:
  \[ 
    \begin{prooftree}
      \Hypo{ p ⊢ Γ, \vec x:A }
      \Infer1{ \bar u(\vec x).p ⊢ Γ, u:↑A }
    \end{prooftree}
  \quad
    \begin{prooftree}
      \Hypo{ p ⊢ Γ, \vec x:A }
      \Infer1{ \bar u(\vec x).p ⊢ Γ, u:?A }
    \end{prooftree}
  \quad
    \begin{prooftree}
      \Hypo { p ⊢ ↑Γ, \vec x:A }
      \Infer1{ u(\vec x).p ⊢ ↑Γ, u:↓A }
    \end{prooftree}
  \quad
    \begin{prooftree}
      \Hypo{ p ⊢ ?Γ, \vec x:A }
      \Infer1{ !u(\vec x).p ⊢ ?Γ, u:!A }
    \end{prooftree}
  \]
  Exchange, contraction and weakening:
  \[
    \begin{prooftree}
      \Hypo{ p ⊢ Γ, \vec x:A, \vec y:B, Δ }
      \Infer1{ p ⊢ Γ, \vec y:B, \vec x:A, Δ }
    \end{prooftree}
  \qquad
    \begin{prooftree}
      \Hypo{ p ⊢ Γ, u:?A, v:?A }
      \Infer1{ p[w/u,v] ⊢ Γ, w:?A }
    \end{prooftree}
  \qquad
    \begin{prooftree}
      \Hypo{ p ⊢ Γ }
      \Infer1{ p ⊢ Γ, u:?A }
    \end{prooftree}
  \]
  Quantifiers:
  \[
    \begin{prooftree}
      \Hypo{ p ⊢ Γ, \vec x:A }
      \Hypo{ X∉\fv(Γ) }
      \Infer2{ p ⊢ Γ, \vec x:∀X.A }
    \end{prooftree}
  \qquad
    \begin{prooftree}
      \Hypo{ p ⊢ Γ, \vec x:A[B/X] }
      \Hypo{ \ar(B)=\ar(X) }
      \Infer2{ p ⊢ Γ, \vec x:∃X.A }
    \end{prooftree}
  \]
  \caption{Typing rules for the \pieq-calculus.}
  \label{table:typing}
\end{table}

We call $\LLa$ (for linear logic with arities) this logical system.
The inference rules are those of multiplicative-exponential linear logic
(MELL), extended with the linear modalities $↑$ and $↓$.
The main difference is in the rule for the existential quantifier: $∃X.A$ can
be deduced from $A[B/X]$ only when $X$ and $B$ have the same arity.
As a consequence, although the language of MELL is a subset of our language of
types, provability of a sequent $Γ$ in MELL is not equivalent to provability
of $Γ$ in $\LLa$.

\subsection{Second-order \lm-calculus} 

Our model of functional computation is the \lm-calculus~\cite{par92:lm}.
We assume an infinite set of \lam-variables ranged over by $x,y$ and an infinite
set of μ-variables ranged over by $α,β$.
Terms are generated by the following grammar:
\[
  M,N ::= x \mid λx.M \mid (M)N \mid μα[β]M
\]
Thus we consider the version of the calculus where $μα$ and $[β]$ cannot
appear separately.
The language of types is minimal second-order logic, i.e.
\[
  A,B ::= X \mid A→B \mid ∀X.A
\]
A typing judgement has the form $Γ⊢M:A\midΔ$ where $Γ$ is a sequence of type
assignments $x:A$ for distinct \lam-variables and $Δ$ is a sequence of type
assignments $α:A$ for distinct μ-variables.
The typing rules are given in table~\ref{table:lm}.
The intuitionistic fragment, i.e. system F, is the fragment of this calculus
where $μα[β]$ is never used and where the $Δ$ part is always empty.

\begin{table}
  Intuitionistic rules:
  \begin{gather*}
    \begin{prooftree}
      \Infer0{ Γ, x:A ⊢ x:A \mid Δ }
    \end{prooftree}
  \qquad
    \begin{prooftree}
      \Hypo{ Γ, x:A ⊢ M:B \mid Δ }
      \Infer1{ Γ ⊢ λx.M:A→B \mid Δ }
    \end{prooftree}
  \\
    \begin{prooftree}
      \Hypo{ Γ ⊢ M:A→B \mid Δ }
      \Hypo{ Γ ⊢ N:A \mid Δ }
      \Infer2{ Γ ⊢ (M)N:B \mid Δ }
    \end{prooftree}
  \end{gather*}
  Quantifiers:
  \[
    \begin{prooftree}
      \Hypo{ Γ ⊢ M:A \mid Δ }
      \Hypo{ X∉\fv(Γ,Δ) }
      \Infer2{ Γ ⊢ M:∀X.A \mid Δ }
    \end{prooftree}
  \qquad
    \begin{prooftree}
      \Hypo{ Γ ⊢ M:∀X.A \mid Δ }
      \Infer1{ Γ ⊢ M:A[B/X] \mid Δ }
    \end{prooftree}
  \]
  Control:
  \[
    \begin{prooftree}
      \Hypo{ Γ ⊢ M:B \mid α:A, β:B, Δ }
      \Infer1{ Γ ⊢ μα[β]M:A \mid β:B, Δ }
    \end{prooftree}
  \]
  \caption{Typing rules for the \lm-calculus.}
  \label{table:lm}
\end{table}

\section{Simply typed \lam-calculus and head linear reduction} 

The basis of linear logic is the decomposition of intuitionistic implication
$A→B$ into an linear implication and an exponential modality, as $!A⊸B$.
The idea is that linear implication $A⊸B=A^⊥⅋B$ is
the actual implication, while the modalities $!A$ and $?A$ control
weakening and contraction.
In this section, we describe the operational meaning of this embedding.
\begin{definition}
  Let $\LJ_0$ be the language of formulas generated by variables and $→$ as
  the only connective.
  The translation $A^0$ of a formula $A$ is defined as
  \begin{align*}
    X^0 &:= X &
    (A→B)^0 &:= !A^0⊸B^0
  \end{align*}
  where each variable of $\LJ_0$ is mapped to a variable of arity $1$ in
  $\LLa$.
\end{definition}

Here propositional variables are considered as base types of arity 1.
A functional type $A_1→\cdots→A_n→X$ is thus translated into a formula of
arity $n+1$.
The translation of formulas naturally induces a translation of type
derivations.
\begin{definition}
  The translation of a simply typed \lam-term $M$ at type $A$ on channels $\vec
  y$ (with $|\vec y|=\ar(A^0)$) is the process $⟦M⟧^A\vec y$ defined as
  \begin{align*}
    ⟦x⟧^A\vec y &:= \bar x⟨\vec y⟩ \\
    ⟦λx.M⟧^{A→B}x\vec y &:= ⟦M⟧^B\vec y \\
    ⟦(M)N⟧^B\vec y &:= \new{x}(⟦M⟧^{A→B}x\vec y | !x(\vec z).⟦N⟧^A\vec z)
  \end{align*}
\end{definition}
The soundness and faithfulness of this translation are easily checked:
\begin{proposition}
  A judgement $x_1:A_1,…,x_n:A_n⊢M:B$ is derivable in $\LJ_0$ if and
  only if $⟦M⟧^B\vec y⊢x_1:?(A_1^0)^⊥,…,x_n:?(A_n^0)^⊥,\vec y:B^0$
  is derivable in $\LLa$.
\end{proposition}

Let us now study the operational meaning of the translation.
Remark that, up to structural congruence, redexes can be permuted
without affecting the translation, i.e. the translation captures
σ-equivalence~\cite{reg90:sigma}.
Subsequently, we get that τ-transitions in the translations correspond to what
is known as head linear reduction~\cite{dr03:pam}.
We briefly recall the definition of these two notions:
\begin{definition}
  σ-equivalence is the congruence over \lam-terms generated by
  \begin{align*}
    (λx.M)NP &=_σ (λx.(M)P)N &
    (λxy.M)N &=_σ λy.(λx.M)N
  \end{align*}
  with $x∉\fv(P)$ and $y∉\fv(N)$.
  Any \lam-term $M$ can be normalised as
  \[
    M =_σ λx_1…x_k(λy_1…y_n.(x)M_1…M_p)N_1…N_n
  \]
  Head linear reduction is the relation over σ-equivalence classes
  generated by
  \begin{multline*}
    λx_1…x_k.(λy_1…y_n.(y_i)M_1…M_p)N_1…N_n \\
    → λx_1…x_k.(λy_1…y_n.(N_i)M_1…M_p)N_1…N_n
  \end{multline*}
\end{definition}

\begin{proposition}
  For any simply typed \lam-term $Γ⊢M:A$, $⟦M⟧^A\vec y$ is bisimilar to $M$ for
  head linear reduction.
\end{proposition}
\begin{proof}
  First note that for $M$ and $N$ of type $A$, if $M=_σN$ then
  $⟦M⟧^A\vec y≡⟦N⟧^A\vec y$, so we can consider terms up to σ-equivalence.
  Consider a typed term $Γ⊢M:A$.
  By σ-equivalence we assume that $M$ is written
  $λ\vec x.(λ\vec y.(x)\vec M)\vec N$ with $|\vec y|=|\vec N|$.
  Call $A_i$ the type of each $x_i$, $B_i$ the type of each $y_i$ and $N_i$
  (these are the same since $M$ is well typed), and call $C_i$ the type of
  each $M_i$.
  Thus we have $A=A_1…A_k→B$ and $x$ has type $C_1…C_p→B$.
  Then we have
  \[\textstyle
    ⟦M⟧^A\vec x\vec z =
    \new{\vec y\vec u}\bigl( \bar x⟨\vec u\vec z⟩
      \bigm| \prod_{i=1}^p ⟦u_i=M_i⟧^{C_i}
      \bigm| \prod_{j=1}^n ⟦y_j=N_j⟧^{B_j} \bigr)
  \]
  with $⟦x=T⟧^A:=!x(\vec y).⟦T⟧^A\vec y$.
  The only possible reduction in this process is $\bar x⟨\vec u\vec z⟩$
  interacting with one of the $⟦u_i=M_i⟧^{C_i}$ or $⟦y_j=N_j⟧^{B_j}$.
  By construction $x$ cannot be one of the $u_i$, so $⟦M⟧^A\vec x\vec z$ has a
  τ-transition if and only if $x=y_j$ for some $j$.
  In this case, we can remark that the following reduction holds:
  \[
    \bar y_j⟨\vec u\vec z⟩ | ⟦y_j=N_j⟧^{B_j} →
    ⟦N_j⟧^{B_j}\vec u\vec z | ⟦y_j=N_j⟧^{B_j}
  \]
  Putting this reduction in context, we get that the reduct of
  $⟦M⟧^A\vec x\vec z$, up to structural congruence, is
  $⟦λ\vec x.(λ\vec y.(N_j)\vec M)\vec N⟧^A\vec x\vec z$.
  Therefore, τ-transitions in translations of \lam-terms strictly correspond to
  head linear reductions in the terms.
\end{proof}

Interestingly, this translation was first described by Hyland and Ong as a
syntax for strategies in a game semantics of PCF~\cite{ho95:dialogue}, thus
with different (but clearly related) arguments.

\section{System F and modal translations} 

The translation presented above is remarkably light.
However, the arity of the translation of a term depends on its type, and as a
consequence polymorphism in the style of system~F does not hold.
Modal translations~\cite{djs95:lktq} are a generalisation of the standard
embedding of intuitionistic logic into linear logic, which allow full
polymorphism by providing a type-independent (and type-safe) translation.
\begin{definition}
  A generalised modality is a word $γ$ over $\{↑,↓,?,!\}$.
  The dual of $γ$ is the modality $\bar{γ}$ such that $(γX)^⊥=\bar{γ}X^⊥$.
  A modal translation of $\LK$ into $\LLa$ is defined by a pair $(γ,δ)$ of
  generalised modalities.
  The translation $A^*$ of a formula $A$ is defined as
  \begin{align*}
    X^* &:= X &
    (A→B)^* &:= γ(A^*)⊸δ(B^*) &
    (∀X.A)^* &:= ∀X.(A^*)
  \end{align*}
  where each variable of $\LK$ is mapped to a variable of arity $2$ in $\LLa$.
  For $Γ=\{x_i:A_i\}_{1≤i≤n}$, define $Γ^*:=\{x_i:A_i^*\}_i$
  and $Γ^{*⊥}:=\{x_i:(A_i^*)^⊥\}_i$.
  For a generalised modality $γ$, define $γΓ:=\{x_i:γA_i\}_i$.
  A type $Γ⊢A|Δ$ is translated at a channel $u$ into the type
  $\bar{γ}Γ^{*⊥},u:δA^*,δΔ^*$.
  A modal translation $(γ,δ)$ is valid if $Γ⊢A\midΔ$ holds
  if and only if $⊢_{\LLa}\bar{γ}Γ^{*⊥},δA^*,δΔ^*$ holds.
\end{definition}

An important fact needs to be stressed: in the source
language $\LK$, any variable can be substituted by any formula.
On the other hand, in the target language $\LLa$,
a variable can only be substituted by a formula of the same arity.
Note that a translation commutes with substitution, i.e.
$(A[B/X])^*=A^*[B^*/X]$, if and only if variables are preserved,
therefore any translation must assign sensible arities to variables.
For this substitution to be correct in any case, we must ensure that the
arity of $A^*$ is independent from $A$.
This condition is satisfied if and only if neither $γ$ nor $δ$ is empty, and
then $\ar(A^*)=2$ for any $A$.

\subsection{General translation} 

\begin{definition}
  Given a non-empty generalised modality $γ$ and names $u$ and $\vec x$,
  define the \emph{protocol} $γu(\vec x).p$ as
  $↓u(\vec x).p := u(\vec x).p$,
  $↑u(\vec x).p := ?u(\vec x).p := \bar u(\vec x).p$,
  and inductively ${γ\dagger}u(\vec x).p:=γu(v).{\dagger}v(\vec x).p$ for
  a fresh name $v$.
  In the case of the empty modality $ε$, let
  $εu(x).p:=p[u/x]$, and
  $εu(\vec x).p$ is undefined for $|\vec x|≠1$.
\end{definition}
Note that protocols are typed in the expected way:
$p⊢Γ,\vec x:A$ implies $γu(\vec x).p⊢Γ,u:γA$.
If $γ$ contains $!$ then the context must be $?Γ$, else if $γ$ contains $↓$
then the context must be $↑Γ$.
For a modal translation $(γ,δ)$ to be valid for classical logic, essentially
two conditions are required:
\begin{itemize}
\item
  It must be possible to apply weakening and contraction to formulas
  $\bar{γ}A$ and $δA$, i.e. $γ$ must start with $!$ and $δ$ must start with
  $?$.
\item
  For the application rule, it must be possible to deduce a common modality
  $ζ$ from $γ$ and $δ$, in a context of $\bar{γ}$ and $δ$ modalities, which
  essentially implies that one of $γ,δ$ must be a suffix of the other.
\end{itemize}
For a pair $(γ,δ)$ to be valid for intuitionistic logic, contraction and
weakening of $δ$ formulas is not required, and contexts only contain $\bar{γ}$
formulas.

Let $(γ,δ)$ be a modal translation for which these conditions are satisfied.
Let $Γ$ and $Δ$ be types where all formulas start with the modalities
$\bar{γ}$ or $δ$.
We can deduce the translation of the rules for $λ$ and $μ$ independently of
the modalities:
\[
  \begin{prooftree}
    \Hypo{ p ⊢ Γ, x:\bar{γ}A^⊥, v:δB }
    \Infer1{ p ⊢ Γ, xv:γA⊸δB }
    \Infer1{ δu(xv).p ⊢ Γ, u:δ(γA⊸δB) }
  \end{prooftree}
\qquad
  \begin{prooftree}
    \Hypo{ p ⊢ Γ, u:δB, α:δA, β:δB }
    \Infer1{ p[β/u] ⊢ Γ, α:δA, β:δB }
  \end{prooftree}
\]
Hence we get
\begin{align*}
  ⟦λx.M⟧u &:= δu(xv).⟦M⟧v &
  ⟦μα[β]M⟧α &:= ⟦M⟧β
\end{align*}
The formulation of the translation of $μα[β]M$ is valid since α-conversion can
be applied to the variable bound by $μ$.
The fact that $μα[β]$ does not modify the process in any other way stresses
the fact that the $μ$ binder is nothing more than a way to name conclusions of
a proof in the sequentialised syntax of \lam-calculus.

It is clear that the introduction rule for $∀$ is not affected by
the translation.
There is a slight difference for the elimination rule: the \lam-calculus is a
syntax for natural deduction with intro/elim, while our type
system for the \pii-calculus is a sequent calculus with only introduction rules
and an actual cut rule.
We can translate the elimination rule for $∀$ by using an extra cut and
axiom:
\begin{prooftree*}
  \Hypo{ p ⊢ Γ, u:∀X.A }
  \Hypo{ u\fuse v ⊢ u:A[B/X]^⊥, v:A[B/X] }
  \Infer1{ u\fuse v ⊢ u:∃X.A^⊥, v:A[B/X] }
  \Infer2{ \new{u}(p|u\fuse v) ⊢ Γ, v:A[B/X] }
\end{prooftree*}
By structural congruence we have $\new{u}(p|u\fuse v)≡p[v/u]$, hence we
can also accept the elimination rule itself in our type system.
For the application rule,
assume there is a generalised modality $ζ$ of which $γ$ and $δ$ are
suffixes, and set $γ'$ and $δ'$ such that $ζ=γ'γ=δ'δ$.
Then the translation of application is:
\begin{prooftree*}
  \small
        \Hypo{ p ⊢ Γ, v:δ(γA⊸δB) }
        \Hypo{ \bar{δ}v⟨xu⟩ ⊢ v:\bar{δ}(γA⊗(δB)^⊥), x:(γA)^⊥, u:δB }
      \Infer2{ \new{v}(p|\bar{δ}v⟨xu⟩) ⊢ Γ, x:(γA)^⊥, u:δB }
    \Infer1{ \bar{γ}'z(x).\new{v}(p|\bar{δ}v⟨xu⟩) ⊢ Γ, z:(ζA)^⊥, u:δB }
      \Hypo{ q ⊢ Δ, w:δA }
    \Infer1{ δ'z(w).q ⊢ Δ, z:ζA }
  \Infer2{ \new{z}(\bar{γ}'z(x).\new{v}(p|\bar{δ}v⟨xu⟩)|δ'z(w).q) ⊢ Γ, Δ, u:δB }
\end{prooftree*}
As explained above, one of $γ,δ$ must be a suffix of the other, so
one of $γ',δ'$ must be empty.
We thus have two cases for the axiom, depending on which one it is:
\[
  \begin{prooftree}
    \Infer0{ u'\fuse u ⊢ u':(δA)^⊥, u:δA }
    \Infer1{ \bar{δ}'x⟨u⟩ ⊢ x:(γA)^⊥, u:δA }
  \end{prooftree}
\qquad
  \begin{prooftree}
    \Infer0{ x\fuse x' ⊢ x:(γA)^⊥, x':γA }
    \Infer1{ γ'u⟨x⟩ ⊢ x:(γA)^⊥, u:δA }
  \end{prooftree}
\]
When both $γ'$ and $δ'$ are empty, these cases collapse into $u\fuse
x⊢x:(γA)^⊥,u:δA$.

\begin{definition}
  Let $(γ,δ)$ be pair of non-empty generalised modalities.
  The translation $⟦M⟧^{γδ}u$ of a \lam-term $M$ is defined inductively by the
  rules of table~\ref{table:general}.
\end{definition}

\begin{table}
  \centering
  $γ$ and $δ$ are given, $γ',δ'$ are such that $γ'γ=δ'δ$.
  \begin{align*}
    ⟦x⟧^{γδ}u &:=
      \begin{cases}
        u\fuse x &\text{if } γ=δ \\
        \bar{δ}'x⟨u⟩ &\text{if } γ=δ'δ \\
        γ'u⟨x⟩ &\text{if } δ=γ'γ
      \end{cases} \\
    ⟦λx.M⟧^{γδ}u &:= δu(xv).⟦M⟧^{γδ}v \\
    ⟦(M)N⟧^{γδ}u &:= \new{z}\Bigl(
        \bar{γ}'z(x).\new{v}\bigl( ⟦M⟧^{γδ}v | \bar{δ}v⟨xu⟩ \bigr) |
        δ'z(w).⟦N⟧^{γδ}w
      \Bigr) \\
    ⟦μα[β]M⟧^{γδ}α &:= ⟦M⟧^{γδ}β
  \end{align*}
  \caption{General case translation of λμ into π.}
  \label{table:general}
\end{table}

\begin{theorem}
  Let $(γ,δ)$ be a valid modal translation.
  For any \lm-term $M$, $Γ⊢M:A\midΔ$ is derivable if and only if
  $⟦M⟧^{γδ}u⊢\bar{γ}Γ^{*⊥},u:δA^*,δΔ^*$ is derivable.
\end{theorem}

Actions in the \pii-calculus, in particular replications, are blocking.
As a consequence, in the standard semantics, there is no reduction
inside replications, so the execution of $⟦M⟧$ does not represent the full
β-reduction.
In the following sections, we give a detailed description of this execution.
As explained above, there are two cases, depending on which of $γ,δ$ is a
suffix of the other:

\begin{definition}
  A pair of generalised modalities $(γ,δ)$ is called left-handed if $δ$ is a
  suffix of $γ$.
  It is called right-handed if $γ$ is a suffix of $δ$.
\end{definition}

\begin{table}
  Classical call-by-name ($γ=!?$, $δ=?$):
  \begin{align*}
    ⟦x⟧u &= \bar x⟨u⟩ \\
    ⟦λx.M⟧u &= \bar{u}(xv).⟦M⟧v \\
    ⟦(M)N⟧u &= \new{v} \bigl( ⟦M⟧v |
      \oc v(xy).( \oc x(w).⟦N⟧w | y \fuse u )
      \bigr)
  \intertext{Intuitionistic call-by-name ($γ=!↓$, $δ=↓$):}
    ⟦x⟧u &= \bar x⟨u⟩ \\
    ⟦λx.M⟧u &= u(xv).⟦M⟧v \\
    ⟦(M)N⟧u &= \new{vx} \bigl(
      ⟦M⟧v | \oc x(w).⟦N⟧w
      | \bar{v}⟨xu⟩ \bigr)
  \intertext{Classical call-by-value ($γ=!$, $δ=?!$):}
    ⟦x⟧u &= \bar{u}⟨x⟩ \\
    ⟦λx.M⟧u &=
      \bar{u}(y).\oc y(xv).⟦M⟧v \\
    ⟦(M)N⟧u &= \new{w} \Bigl(
      \oc w(x).\new{v} \bigl(
        ⟦M⟧v | \oc v(w).\bar w⟨xu⟩ \bigr)
      \Bigm| ⟦N⟧w \Bigr)
  \intertext{Intuitionistic call-by-value ($γ=!$, $δ=!$):}
    ⟦x⟧u &= u\fuse x \\
    ⟦λx.M⟧u &=
      \oc u(xv).⟦M⟧v \\
    ⟦(M)N⟧u &= \new{vw} \bigl(
        ⟦M⟧v | ⟦N⟧w |
        \bar{v}⟨wu⟩ \bigr)
  \end{align*}
  \caption{Particular cases of translations.}
  \label{table:translations}
\end{table}

\subsection{Call-by-name} 

Here we consider the left-handed case, i.e. with $γ=δ'δ$ for some
non-empty $δ'$.
As a simplification we consider the case where $δ$ and $δ'$ are simple
modalities, one easily checks that the other cases
are not significantly different.
The validity constraints impose $δ'=!$, and $δ$ has to be $?$ for the
classical case.

To describe precisely the operation of translated terms, we introduce a new
form of term $♯M$ and define a continuation $K$ as $M_1…M_kα$ where $α$ is a
μ-variable and the $M_i$ are terms.
An executable is a pair $M*K$,
equivalence $≡$ and execution $→$ of executables are defined as
\begin{align*}
  (M)N * K &≡ M * N K &
  ♯M * K &→ M * K \\
  μα[β]M * K &≡ M[K/α] * β &
  λx.M * N K &→ M[♯N/x] * K
\end{align*}
The substitution $M[M_1…M_nα/β]$ is the substitution of every
subterm of the form $[β]N$ of $M$ by $[α](N)M_1…M_n$.
The translation of terms is extended to executables as
\begin{align*}
  ⟦α⟧u &:= α\fuse u &
  ⟦x=M⟧ &:= δ'x(u).⟦M⟧u \\
  ⟦MK⟧u &:= \new{vz}(\bar{δ}u⟨zv⟩ | ⟦z=M⟧ | ⟦K⟧v) &
  ⟦α=K⟧ &:= ⟦K⟧α \\
  ⟦♯M⟧u &:= \new{x}(⟦x⟧u | ⟦x=M⟧) &
  ⟦M*K⟧ &:= \new{u}(⟦M⟧u|⟦K⟧u)
\end{align*}

\begin{proposition}
  For any call-by-name executables $e_1$ and $e_2$,
  $e_1≡e_2$ implies $⟦e_1⟧≅⟦e_2⟧$ and
  $e_1→e_2$ if and only if $⟦e_1⟧→⟦e_2⟧$.
\end{proposition}
\begin{proof}
  Remark that the translation $⟦M⟧u$ of a variable or an abstraction has
  exactly one transition, labelled by an action on $u$ or on a variable.
  Similarly, the translation $⟦K⟧u$ of a continuation either is an equator
  $u\fuseα$ or has a unique transition labelled by an action on $u$.
  $⟦♯M⟧u$ has a single transition to a process bisimilar to $⟦M⟧u$.
  Then the key of the proof is the remark that bindings correctly
  implement substitution up to bisimilarity, i.e.
  $\new{α}(⟦e⟧|⟦α=K⟧)≅⟦e[K/α]⟧$ for any fresh name $α$, and
  $\new{x}(⟦e⟧|⟦x=M⟧)≅⟦e[♯M/x⟧$ for any fresh name $x$.
  The rule for $μα[β]$ applies only in the classical case, then $\bar{δ}$
  starts with $!$ and continuations are replicable.
  Details can be found in the appendix.
\end{proof}

Executing a \lm-term simply means executing it on a continuation $α$
for a fresh variable $α$, since $⟦M⟧α≡⟦M*α⟧$.
Hence we can summarise this result as:
\begin{theorem}
  Left-handed translations implement call-by-name execution.
\end{theorem}

The case for $γ=!↓$ and $δ=↓$ is an adaptation of the standard $!A⊸B$
decomposition that allows polymorphism.
Operationally, it exactly corresponds to Milner's
translation~\cite{mil90:fun}.
The case for $γ=!?$ and $δ=?$ corresponds to the system known as LKT in
Danos-Joinet-Schellinx.
As far as we know, its operational counterpart in the \pii-calculus is new.
These particular translations are shown in table~\ref{table:translations}.
In the classical case, the application uses an equator $y\fuse u$ which is not
standard \pii-calculus, however it can be argued that replacing it by a forwarder
$!y(ab).\bar u⟨ab⟩$ does not affect the validity of the translation, although
the step-by-step operational description is a bit heavier to formulate.

\subsection{Call-by-value} 

We now consider the right-handed case, i.e. with $δ=γ'γ$.
As in the previous section, we assume without loss of generality that $γ$ is a
single modality, necessarily $!$ because of the validity constraints.
We now have two main choices for $γ'$, namely $?$ for the classical case and
$↑$ for the intuitionistic case.
We now have to distinguish values, terms and continuations:
\begin{align*}
  &\text{values}&
  V,W &:= x \mid λx.V \\
  &\text{terms}&
  M,N &:= V \mid (M)N \mid μα[K]M \mid V⋅W \\
  &\text{continuations}&
  K,L &:= α \mid K M^f \mid K V^a
\end{align*}
An executable is a pair $K*M$.
Equivalence and execution are defined as
\begin{align*}
  &&
  K M^f * V &→ K V^a * M \\
  K * (M)N &≡ K M^f * N &
  K W^a * V &→ K * V⋅W \\
  K * μα[L]M &≡ L * M[K/α] &
  K * λx.M⋅V &→ K * M[V/x]
\end{align*}
A continuation contains functions as unevaluated terms $M^f$
and arguments as values $V^a$, so arguments are evaluated first.
The terms $V⋅W$ and $μα[K]M$ are introduced to get a precise
bisimulation.
Translations are extended as
\begin{align*}
  ⟦V⟧u &:= γ'u(x).⟦x=V⟧ \\
  ⟦(M)N⟧u &:= \new{v}( ⟦v=uM^f⟧ | ⟦N⟧v ) \\
  ⟦V⋅W⟧u &:= \new{xy}( ⟦x=V⟧ | ⟦y=W⟧ | \bar{γ}x⟨yu⟩ ) \\
  ⟦x=y⟧ &:= x\fuse y \\
  ⟦x=λy.M⟧ &:= γx(yu).⟦M⟧u \\
  ⟦α=β⟧ &:= α\fuse β \\
  ⟦α=KM^f⟧ &:= \new{v}( \bar{γ'}α(x).\new{u}( ⟦M⟧u | \bar{δ}u⟨xv⟩ ) | ⟦v=K⟧ ) \\
  ⟦α=KV^a⟧ &:= \new{vx}( ⟦x=V⟧ | \bar{δ}α⟨xv⟩ | ⟦v=K⟧ ) \\
  ⟦μα[K]M⟧α &:= \new{β}( ⟦β=K⟧ | ⟦M⟧β ) \\
  ⟦K*M⟧ &:= \new{u}( ⟦u=K⟧ | ⟦M⟧u )
\end{align*}
\begin{proposition}
  For any call-by-value executables $e_1$ and $e_2$,
  $e_1≡e_2$ implies $⟦e_1⟧≅⟦e_2⟧$ and
  $e_1→e_2$ if and only if $⟦e_1⟧→⟦e_2⟧$.
\end{proposition}
\begin{proof}
  The proof follows the same principle as in call-by-name.
  The substitution lemma now states $\new{x}(⟦e⟧|⟦x=V⟧)≅⟦e[V/x]⟧)$ where $e$
  is an executable, $V$ is a value and $x$ is a \lam-variable; the same lemma for
  μ-variables and continuations also holds.
  We then remark that translations of terms and continuations always have at
  most one transition, and the correspondence with the operational semantics
  above is easily checked.
  Details can be found in the appendix.
\end{proof}

Given a fresh μ-variable $α$, once again we get $⟦M⟧α=⟦α*M⟧$, hence the
semantics above precisely describes the execution of translations of \lm-terms
in right-handed translations, which can be summarised as follows:
\begin{theorem}
  Right-handed translations implement call-by-value execution.
\end{theorem}

The case for $γ=!$ and $δ=?!$ corresponds to the system called LKQ in
Danos-Joinet-Schellinx.
Operationally, we get exactly Honda, Yoshida and Berger's
translation~\cite{bhy03:genepi,hyb04:contpi}.
The case for $γ=!$ and $δ=↑!$ is a version of this translation linearised with
respect to conclusions.
It is actually very close to Milner's encoding of call-by-value
\lam-calculus~\cite{mil90:fun}, which corresponds to the slightly more expensive
decomposition $(A→B)^*=↓(!A^*⊸↑!B^*)$.

The simplest intuitionistic version is obtained by taking $γ=δ=!$, which
is both left- and right-handed.
It is easy to check that the operational meaning of this translation is an
extension of the call-by-value strategy where functions and arguments can be
executed in parallel.
These translations are shown in table~\ref{table:translations}.

\section{Realisability interpretations} 

The previous sections define a family of type-preserving translations of the
\lm-calculus into the \pieq-calculus, and provide a detailed description of the
operational semantics induced by the translations.
Since the operational translations are deduced from simple embeddings of
intuitionistic and classical logics into linear logic, we can expect more
semantic interpretations.

The soundness of the type system we use for processes is formulated using
realisability, as described in the following section.

\subsection{Soundness of $\LLa$} 

For a finite set of names $I$, a process $p$ has interface $I$ if $\fv(P)⊆I$.
\begin{definition}
  An observation is a set $⊥$ of processes of empty interface.
  Given an observation $⊥$, two processes $p$ and $q$ of interface $I$
  are orthogonal, written $p\perp q$, if $\new{I}(p|q)∈⊥$.
  An observation $⊥$ is valid if
  \begin{itemize}
  \item $⊥$ is closed under bisimilarity,
  \item if $p$ has a unique labelled transition $p\xrightarrow{τ}p'$ and
    $p'\perp q$ then $p\perp q$.
  \end{itemize}
\end{definition}
If $\?A$ is a set of processes of interface $I$, its orthogonal is
the set $\?A^⊥:=\set{p:I}{∀q∈\?A,p\perp q}$.
A behaviour is a set $\?A$ such that $\?A=\?A^{⊥⊥}$.
The complete lattice of behaviours of interface $I$ is noted $\Beh_I$.

Let $(u_i)_{i∈\NN}$ be an infinite sequence of pairwise distinct names.
Let $\Beh_k:=\Beh_{u_1…u_k}$.
A valuation of propositional variables is a function $ρ$ that associates, to
each variable $X$ of arity $k$, a behaviour $ρ(X)∈\Beh_k$.
Given a valuation $ρ$, the interpretation of a type $A$ localised at $\vec x$,
with $|\vec x|=\ar(A)$, is the behaviour $⟦\vec x:A⟧ρ$ of interface $\vec x$
defined inductively by
  \begin{align*}
    ⟦x_1…x_n:X⟧ρ &:= v(X)[x_1/u_1,…,x_n/u_n]
\\  ⟦\vec x\vec y:A⊗B⟧ρ &:= \set{(p|q)}{p∈⟦\vec x:A⟧ρ,q∈⟦\vec y:B⟧ρ}^{⊥⊥}
\\  ⟦u:↓A⟧ρ &:= \set{u(\vec x).p}{p∈⟦\vec x:A⟧ρ}^{⊥⊥}
\\  ⟦\vec x:∃X^k.A⟧ρ &:= \textstyle
      \bigl(\bigcup_{\?X∈\Beh_k}⟦\vec x:A⟧(ρ[X:=\?X]) \bigr)^{⊥⊥}
  \end{align*}
and $⟦\vec x:A^⊥⟧ρ:=(⟦\vec x:A⟧ρ)^⊥$.
Exponential modalities require a more subtle definition:
for each name $u$, define the contraction $δ_u$ over behaviours of
interface $\{u\}$ as
\[
  δ_u(\?A) := \set{ p[u/v,w] }{ p ∈ \?A[v/u] ⅋ \?A[w/u] }^{⊥⊥}
\]
where $v$ and $w$ are fresh names.
Then, for a behaviour $\?B$ of interface $\{x_1…x_n\}$, define
$
  F_u(\?B,\?X) := \bigl( ⟦u:↑\?B⟧ ∪ \{1:u\}^⊥ ∪ δ_u(\?X) \big)^{⊥⊥}
$.
This operator is obviously monotonic in $\?X$, and the interpretation of
exponential modalities is defined as a fixed point of it:
\begin{align*}
  ⟦u:?A⟧ρ &:= \operatorname{lfp}(\?X \mapsto F_u(⟦\vec x:A⟧ρ, \?X)) &
  ⟦u:!A⟧ρ &:= (⟦u:?(A^⊥)⟧ρ)^⊥
\end{align*}
Finally, a type $Γ=\vec x_1:A_1,…,\vec x_n:A_n$ is interpreted as
\[
  ⟦Γ⟧ρ :=
  \set{ (p_1|…|p_n) }{ p_1∈⟦\vec x_1:A_1⟧ρ^⊥, …, p_n∈⟦\vec x_n:A_n⟧ρ^⊥ }^⊥
\]
\begin{definition}
  Given an observation, a process $p$ realises a type $Γ$ if $p∈⟦Γ⟧ρ$ for any
  valuation $ρ$.
  This fact is written $p⊩Γ$.
\end{definition}
From the definition of observations and the interpretation of formulas, we
easily deduce the adequacy theorem (we do not expose the proof here, a
detailed study on this technique can be found in other works by the
author~\cite{bef05:lrc,bef05:cmll}):
\begin{theorem}
  If $p⊢Γ$ is derivable, then $p⊩Γ$ for any observation $⊥$.
\end{theorem}

The usual notions of testing fit in our notion of observation, for instance:
\begin{proposition}
  Let $ω$ be a channel,
  assume $ω$ is not taken into account in interfaces.
  Define the must-testing observation as $\set{p}{∀p→^*q,∃q→^*ω|r}$.
  Must-testing is a valid observation.
\end{proposition}
Properties of typed processes, such as termination or deadlock-freeness, can
be obtained by choosing appropriate observations. For instance:
\begin{proposition}\label{prop:deadlock}%
  Let $p⊢Γ$ be a typed process such that any propositional variable occurring
  in $Γ$ is under a modality.
  For any reduction $p→^*p'$ there is a reduction $p'→^*p''$ such that $p''$
  has a visible action.
\end{proposition}
\begin{proof}
  We use the must-testing observation with a channel $ω$ that does
  not occur in $p$.
  Note that $ω∈⟦\vec x:A⟧$ for any formula $A$, hence $u(\vec x).ω∈⟦u:↓A⟧$.
  By similar arguments we get $u(\vec x).ω∈⟦u:!A⟧$, $\bar u(\vec x).ω∈⟦u:↑A⟧$
  and $\bar u(\vec x).ω∈⟦u:?A⟧$.
  Moreover it is clear that, for $q∈⟦A⟧$ and $r∈⟦B⟧$,
  $(q|r)∈⟦A⊗B⟧$ and $(q|r)∈⟦A⅋B⟧$.
  Each name $u_i$ occurring in $Γ$ occurs with a polarity $ε_i$ (depending on
  the modality that introduces it) and a particular arity.
  Let $t:=\prod_i\bar u_i^{ε_i}(\vec x).ω$, by the above remarks we know
  that $t∈(⟦Γ⟧ρ)^⊥$ for any valuation $ρ$.
  This implies that, for any reduction $p|t→^*p'|t$ there is a reduction
  $p'|t→^*ω|q$.
  Since $ω$ only occurs in $t$, this implies that an action in $t$
  must be triggered during this reduction.
  By induction on the typing rules, on proves that if all type variables
  occur under modalities, no equator in $p$ can relate free names,
  hence triggering an action in $t$ must be done by an action in
  a reduct of $p'$.
\end{proof}

\begin{corollary}
  The execution of a typed \lm-term in call-by-name or call-by-value always
  ends with a \lam- or μ-variable in active position.
\end{corollary}
\begin{proof}
  Let $Γ⊢M:A|Δ$ be a typed \lm-term.
  Using non-divergence as the observation we can prove that $⟦M⟧α$ has no
  infinite reduction.
  Consider a reduction $⟦M⟧α→^*p$ with $p$ irreducible.
  By proposition~\ref{prop:deadlock} we deduce that $p$ must have a visible
  action, and this action can only be on $α$ or a name that occurs in $Γ$ or
  $Δ$.
  Conclude by reasoning on the shape of translations of terms: in
  call-by-name, executables with visible actions are $x*K$ or $λx.M*α$; in
  call-by-value they are $K*x⋅V$ or $α*λx.M$.
\end{proof}

\subsection{Extending the \lm-calculus} 

Realisability presents the type system $\LLa$ as an axiomatisation of the
algebra of process behaviours.
This allows for the introduction of new logical connectives and new rules: by
semantic means (i.e. by reasoning on the reductions of processes) we can
define the interpretation of a connective as an operation on sets of
processes.
If we prove the adequacy of a new logical rule, we can then use it as a typing
rule for processes with the guarantee that any property that is proved by
realisability is preserved; this includes termination and deadlock-freeness.

This technique can be used to extend the typed \lm-calculus.
As soon as a connective can be translated into $\LLa$ (possibly extended as
explained above), a translation of the underlying syntax is deduced
the same way as for the core calculus, which induces an evaluation
strategy.
This provides a framework for extending our type-preserving translations,
without loosing any of the properties of the translations.
We now provide some examples of these ideas.

\paragraph{Product types}

Products can be added to the \lm-calculus by means of a pair of constructs for
introduction and elimination:
\[
  \begin{prooftree}
    \Hypo{ Γ ⊢ M:A | Δ }
    \Hypo{ Γ ⊢ N:B | Δ }
    \Infer2{ Γ ⊢ (M,N):A×B | Δ }
  \end{prooftree}
\quad
  \begin{prooftree}
    \Hypo{ Γ ⊢ M:A×B | Δ }
    \Hypo{ Γ, x:A, y:B ⊢ N:C | Δ }
    \Infer2{ Γ ⊢ \LET{x,y=M}{N}:C | Δ }
  \end{prooftree}
\]
Given a pair $(γ,δ)$, we extend the translation of types by
$(A×B)^*=γA^*⊗γB^*$.
Note that, when $γ$ and $δ$ are not empty, the arity of $(A⊗B)^*$ is $2$,
hence polymorphism is preserved.
The translation of terms is extended as follows:
\begin{align*}
  ⟦(M,N)⟧^{γδ}u &:= δu(xy).( δ'x(v).⟦M⟧^{γδ}v | δ'y(w).⟦N⟧^{γδ}w ) \\
  ⟦\LET{x,y=M}{N}⟧^{γδ}u &:= \new{v}( ⟦M⟧^{γδ}v | \bar{δ}v(xy).⟦N⟧^{γδ}u )
\end{align*}
In both strategies, $\LET{x,y=M}{N}$ must reduce
$M$ into a pair before evaluating $N$.
The evaluation of the parts of a pair in call-by-value is done
in parallel since $δ'$ is empty.
We leave to the reader the formulation of precise evaluation rules.

\paragraph{Sum types}

Sum types in λμ can be defined as follows (with $i∈\{1,2\}$):
\begin{gather*}
  \begin{prooftree}
    \Hypo{ Γ ⊢ M:A_i | Δ }
    \Infer1{ Γ ⊢ \inj_iM:A_1+A_2 | Δ }
  \end{prooftree}
\quad
  \begin{prooftree}
    \Hypo{ Γ ⊢ M:A_1+A_2 | Δ }
    \Hypo{ Γ, x_i:A_i ⊢ N_i:C | Δ }
    \Infer2{ Γ ⊢ \CASE{M}{\inj_ix_i→N_i}:C | Δ }
  \end{prooftree}
\end{gather*}
Decomposing this in linear logic requires the additives $⊕$ and $\with$.
The general rules in $\LLa$ are complicated, but here we only need
simplified versions:
\[
  \begin{prooftree}
    \Hypo{ p ⊢ Γ, u:↑A }
    \Infer1{ p ⊢ Γ, uv:↑A⊕↑B }
  \end{prooftree}
\qquad
  \begin{prooftree}
    \Hypo{ p ⊢ Γ, \vec x:A }
    \Hypo{ q ⊢ Γ, \vec y:B }
    \Infer2{ u(\vec x).p+v(\vec y).q ⊢ Γ, uv:↓A\with↓B }
  \end{prooftree}
\]
assuming the underlying \pii-calculus has guarded choice.
We get adequacy by defining $⟦uv:A⊕B⟧ρ:=(⟦u:A⟧ρ∪⟦v:B⟧ρ)^{⊥⊥}$ and interpreting
$A\with B$ by duality.
The sum type of λμ is translated as $(A+B)^*=↑γA^*⊕↑γB^*$
(which preserves polymorphism).
The translation of terms follows:
\begin{align*}
  ⟦ \inj_iM ⟧u &:= δu(a_1a_2).↑δ'a_i(v).⟦M⟧v \\
  ⟦ \CASE{M}{\inj_ix_i→N_i} ⟧u &:= \textstyle \new{v} \bigl(
    ⟦M⟧v \bigm| \bar{δ}v(ab).\sum_i a_i(x_i).⟦N_i⟧u \bigr)
\end{align*}
Obviously, in any strategy, the evaluation of $\CASE{M}{\inj_ix_i→N_i}$ must
always reduce $M$ into an $\inj_i$ before proceeding.

\paragraph{Subtyping}

Behaviours of a given interface form a complete lattice, with intersection as
the lower bound and bi-orthogonal of the union as the upper bound.
Write $∧$ and $∨$ these dual connectives with $\ar(A∧B)=\ar(A)=\ar(B)$.
This induces subtyping over types, defined as $A≤B$ if $⟦A⟧⊆⟦B⟧$, and the
rules:
\[
  \begin{prooftree}
    \Hypo{ p ⊢ Γ, \vec x:A }
    \Hypo{ p ⊢ Γ, \vec x:B }
    \Infer2{ p ⊢ Γ, \vec x:A∧B }
  \end{prooftree}
\quad
  \begin{prooftree}
    \Hypo{ p ⊢ Γ, \vec x:A }
    \Infer1{ p ⊢ Γ, \vec x:A∨B }
  \end{prooftree}
\quad
  \begin{prooftree}
    \Hypo{ p ⊢ Γ, \vec x:A }
    \Hypo{ A ≤ B }
    \Infer2{ p ⊢ Γ, \vec x:B }
  \end{prooftree}
\]
It is clear that all connectives except negation are
increasing for this relation, and that $A≤B$ if and only if $B^⊥≤A^⊥$.
By the interpretation of modalities we also get $!A≤↓A$ and $↑A≤?A$.
Subtyping rules in λμ can be written as
\[
  \begin{prooftree}
    \Hypo{ Γ ⊢ M:A | Δ }
    \Hypo{ Γ ⊢ M:B | Δ }
    \Infer2{ Γ ⊢ M:A∩B | Δ }
  \end{prooftree}
\qquad
  \begin{prooftree}
    \Hypo{ Γ ⊢ M:A | Δ }
    \Hypo{ A ≤ B }
    \Infer2{ Γ ⊢ M:B | Δ }
  \end{prooftree}
\]
Translations are extended as $(A∩B)^*=A^*∧B^*$.
The usual subtyping rules, like $(A→B)≤(A'→B')$ if $A'≤A$ and $B≤B'$,
hold through translation.

\paragraph{Fix points}

The fact that behaviours form complete lattices also guarantees that any
increasing function over behaviours of a fixed interface have (least and
greatest) fix points.
We can thus extend $\LLa$ with dual constructs $μX.A$ and $νX.A$, with the
constraints that $\ar(X)=\ar(A)$ and that $X$ does not occur as $X^⊥$ in $A$.
The typing rules for fix points are rather technical to formulate, mainly
because the proper rule for $νX.A$ requires the introduction of a recursion
operator in the \pii-calculus.
Fix points in the types for \lm-calculus would be simply translated as
$(μX.A)^*=μX.(A^*)$.
The constraint that permits polymorphism à la system F also allows this fix
point to be used for any $A$ where $X$ only occurs positively.

\bigskip

These various extensions to the type system can be freely combined.
Other extensions, notably with concurrent primitives, could be studied in a
similar way.
However, for this purpose, it seems necessary to enforce serious linearity in
the calculus.
This fits naturally in our type system for the \pii-calculus but it
is incompatible with full control in the style we get from
translations of full classical logic.
Precise studies of this idea are deferred to further work.


\bibliographystyle{ebstyle}
\bibliography{fpp}

\vfill\pagebreak
\appendix
\section{Technical details} 

\subsection{Bisimulation in \pieq} 

A polarity $ε$ is an element of $\{↓,↑\}$.
$↓$ is called positive and $↑$ is called negative.
The notation $u^ε(\vec x)$ stands for $u(\vec x)$ if $ε=↓$ and for
$\bar u(\vec x)$ if $ε=↑$.

Two names $x$ and $y$ are unified by a process $p$ if $p⊨x=y$ is
derivable using the rules of table~\ref{table:unif}.
Note that an action like $u(x).y\fuse z$ does not unify $y$ and $z$, i.e.
the equator $y\fuse z$ is inactive as long as the action $u(x)$ has not been
consumed.
A transition can have one of three kinds of labels:
\begin{align*}
  e ::= {}&
  u^ε(x_1…x_n) &&\text{visible action (with the $x_i$ fresh and distinct)} \\
  &[u\fuse v] &&\text{conditional internal reduction} \\
  &τ &&\text{internal reduction}
\end{align*}
The notation $p⊨a=b$ is extended to transition labels as detailed in
table~\ref{table:unif}.
For a label $e$, $\names(e)$ is the set of names that occur in $e$, i.e.
$\names(u(x_1…x_n))=\{u,x_1…x_n\}$, $\names([u\fuse v])=\{u,v\}$ and
$\names(τ)=∅$.
The labelled transition system of the calculus is defined in
table~\ref{table:lts}.
\begin{table}
  Axioms and context rules for unification:
  \[
    \begin{prooftree}
      \Infer0{ x\fuse y ⊨ x=y }
    \end{prooftree}
  \qquad
    \begin{prooftree}
      \Hypo{ p ⊨ x=y }
      \Infer1{ p|q ⊨ x=y }
    \end{prooftree}
  \qquad
    \begin{prooftree}
      \Hypo{ p ⊨ x=y }
      \Infer1{ q|p ⊨ x=y }
    \end{prooftree}
  \qquad
    \begin{prooftree}
      \Hypo{ p ⊨ x=y }
      \Hypo{ z∉\{x,y\} }
      \Infer2{ \new{z}p ⊨ x=y }
    \end{prooftree}
  \]
  Reflexivity, symmetry and transitivity of equators:
  \[
    \begin{prooftree}
      \Infer0{ p ⊨ x=y }
    \end{prooftree}
  \qquad
    \begin{prooftree}
      \Hypo{ p ⊨ x=y }
      \Infer1{ p ⊨ y=x }
    \end{prooftree}
  \qquad
    \begin{prooftree}
      \Hypo{ p ⊨ x=y }
      \Hypo{ p ⊨ y=z }
      \Infer2{ p ⊨ x=z }
    \end{prooftree}
  \]
  Renaming of transition labels:
  \[
    \begin{prooftree}
      \Hypo{ p ⊨ u=v }
      \Hypo{ u,v∉\{x_1…x_n\} }
      \Infer2{ p ⊨ u^ε(x_1…x_n)=v^ε(x_1…x_n) }
    \end{prooftree}
  \qquad
    \begin{prooftree}
      \Hypo{ p ⊨ u=u' }
      \Hypo{ p ⊨ v=v' }
      \Infer2{ p ⊨ [u\fuse v]=[u'\fuse v'] }
    \end{prooftree}
  \]
  \caption{Rules for name unification.}
  \label{table:unif}
\end{table}
\begin{table}
  Actions (with $α=u^ε(x_1…x_n)$) and composition:
  \[
    \begin{prooftree}
      \Infer0{ α.p \xrightarrow{α} p }
    \end{prooftree}
  \qquad
    \begin{prooftree}
      \Infer0{ !α.p \xrightarrow{α} p | !α.p }
    \end{prooftree}
  \qquad
    \begin{prooftree}
      \Hypo{ p \xrightarrow{\bar u(x_1…x_n)} p' }
      \Hypo{ q \xrightarrow{v(x_1…x_n)} q' }
      \Infer2{ p|q \xrightarrow{[u\fuse v]} \new{x_1…x_n}(p'|q') }
    \end{prooftree}
  \]
  Renaming:
  \[
    \begin{prooftree}
      \Hypo{ p \xrightarrow{e} p' }
      \Hypo{ p ⊨ e=e' }
      \Infer2{ p \xrightarrow{e'} p' }
    \end{prooftree}
  \qquad
    \begin{prooftree}
      \Hypo{ p \xrightarrow{[u\fuse v]} p' }
      \Hypo{ p ⊨ u=v }
      \Infer2{ p \xrightarrow{τ} p' }
    \end{prooftree}
  \]
  Context:
  \[
    \begin{prooftree}
      \Hypo{ p \xrightarrow{e} p' }
      \Infer1{ p|q \xrightarrow{e} p'|q }
    \end{prooftree}
  \qquad
    \begin{prooftree}
      \Hypo{ p \xrightarrow{e} p' }
      \Infer1{ q|p \xrightarrow{e} q|p' }
    \end{prooftree}
  \qquad
    \begin{prooftree}
      \Hypo{ p \xrightarrow{e} p' }
      \Hypo{ x∉\names(e) }
      \Infer2{ \new{x}p \xrightarrow{e} \new{x}p' }
    \end{prooftree}
  \]
  \caption{Labelled transition system.}
  \label{table:lts}
\end{table}

A simulation is a relation $\?S$ over processes such that $p\?S q$ implies
that
\begin{itemize}
\item
  for any $x,y∈\Names$, $p⊨x=y$ implies $q⊨x=y$,
\item 
  for each transition $p\xrightarrow{e}p'$ there is a transition
  $q\xrightarrow{e}q'$ such that $p'\?S q'$.
\end{itemize}
A bisimulation is a relation $\?S$ such that both $\?S$ and $\?S^{-1}$ are
simulations.
Two processes $p$ and $q$ are bisimilar if there is a bisimulation $\?S$ such
that $p\?S q$.

\subsection{Simulation in call-by-name} 

For the “push” rule, we have:
\begin{align*}
  ⟦ (M)N * K ⟧
  &= \new{u}
    \Bigl( \new{z}
      \bigl( \new{v} ( ⟦M⟧v | \bar{δ}v⟨zu⟩ )
      \bigm| ⟦z=N⟧ \bigr)
    \Bigm| ⟦K⟧u
    \Bigr) \\
  &≡ \new{uvz}
    \bigl( ⟦M⟧v
    \bigm| \bar{δ}v⟨zu⟩
    \bigm| ⟦z=N⟧
    \bigm| ⟦K⟧u
    \bigr) \\
  &≡ \new{v}
    \Bigl( ⟦M⟧v
    \Bigm| \new{u}
      \bigl( \new{z} ( \bar{δ}v⟨zu⟩ | ⟦z=N⟧ )
      \bigm| ⟦K⟧u \bigr)
    \Bigr) \\
  &≡ \new{v}
    \bigl( ⟦M⟧v
    \bigm| \new{uz} ( ⟦v=Nu⟧ | ⟦K⟧u )
    \bigr) \\
  &= ⟦ M * NK ⟧
\end{align*}
For the substitution rule for continuations, consider a process
$\new{α}(⟦M⟧u|⟦K⟧α)$, with $K=M_1…M_kβ$.
When $δ$ starts with $?$, each $⟦K⟧α$ is a guarded replication on channel $α$.
By construction there is no other input on $α$ so each output on $α$ can only
interact with $⟦K⟧α$.
Hence, up to bisimilarity, we can distribute $⟦K⟧α$ in $⟦M⟧u$ by substituting
each action $\bar{α}(\vec x).p$ by $\new{α'}(\bar{α}'(\vec x).p|⟦K⟧α')$ for a
fresh $α'$.
All output occurrences of $α$ occur in processes of the form $⟦μθ[α]N⟧θ=⟦N⟧α$,
but $\new{α'}(⟦N⟧α'|⟦K⟧α')=⟦N*K⟧$ and by the previous rule we have
$⟦N*K⟧≡⟦(N)M_1…M_k*β⟧=⟦μθ[β](N)M_1…M_k⟧θ$.
By this rule we can deduce the validity of the rule for $μα[β]$:
\[
  ⟦ μα[β]M * K ⟧
  = \new{α}
    \bigl( ⟦M⟧β
    \bigm| ⟦K⟧α
    \bigr) \\
  ≅ ⟦ M[K/α]  ⟧β
  ≡ ⟦ M[K/α] * β ⟧
\]
In the intuitionistic case the rule is not applicable, but it would hold too
under the condition that each μ-variable is used linearly.
For the $♯M$ rule, we have:
\begin{align*}
  ⟦ ♯M * K ⟧
  &= \new{ux}
    \bigl( \bar{δ'}x⟨u⟩
    \bigm| δ'x(v).⟦M⟧v
    \bigm| ⟦K⟧u
    \bigr) \\
  &→ \new{ux}
    \bigl( ⟦M⟧u
    \bigm| ⟦K⟧u
    \bigm| δ'x(v).⟦M⟧v \bigr) \\
  &≅ \new{u} \bigl( ⟦M⟧u \bigm| ⟦K⟧u \bigr)
  = ⟦ M * K ⟧
\end{align*}
where $→$ contains one transition for each modality in the word $δ'$.
Since $⟦K⟧u$ and $E$ are blocked on actions that cannot be on channel $x$,
this reduction is clearly the only one possible.
The term $δ'x(v).⟦M⟧v$ is not consumed since $δ'$ must contain $!$, however
there is no other occurrence of $x$ so we can discard it by bisimilarity.

For the substitution rule for terms, the argument is the same as for
continuations.
In this case, the only outputs on the channel of a \lam-variable $x$ are of the
form $⟦x⟧u=\bar{δ}'x⟨u⟩$, hence after distribution of $⟦x=M⟧$ we get
$\new{x'}(⟦x'⟧u|⟦x'=M⟧)=⟦♯M⟧u$ for a fresh $x'$.
For the “pop” rule, we thus have
\begin{align*}
  ⟦ λx.M * NK ⟧
  &= \new{u}
    \Bigl( δu(xv).⟦M⟧v
    \Bigm| \new{w}
      \bigl( \new{z} ( \bar{δ}u⟨zw⟩ | ⟦z=N⟧ )
      \bigm| ⟦K⟧w \bigr)
    \Bigr) \\
  &≡ \new{uwx}
    \bigl( δu(xv).⟦M⟧v
    \bigm| \bar{δ}u⟨xw⟩
    \bigm| ⟦x=N⟧
    \bigm| ⟦K⟧w
    \bigr) \\
  &→ \new{uwx}
    \bigl( ⟦M⟧w
    \bigm| \bar{δ}u⟨xw⟩
    \bigm| ⟦x=N⟧
    \bigm| ⟦K⟧w
    \bigr) \\
  &≡ \new{wx}
    \bigl( ⟦M⟧w
    \bigm| \new{u} \bar{δ}u⟨xw⟩
    \bigm| ⟦x=N⟧
    \bigm| ⟦K⟧w
    \bigr) \\
  &≅ \new{wx}
    \bigl( ⟦M⟧w
    \bigm| ⟦x=N⟧
    \bigm| ⟦K⟧w
    \bigr) \\
  &≅ \new{w}
    \bigl( ⟦M[♯N/x]⟧w
    \bigm| ⟦K⟧w
    \bigr) \\
  &= ⟦ M[♯N/x] * K ⟧
\end{align*}
where $→$ contains one transition for each modality in the word $δ$.
In the classical case, $\bar{δ}u⟨xw⟩$ is not consumed since $\bar{δ}$ contains
$!$, however we know that $u$ does not occur elsewhere since all duplications
of continuations are performed by the rule for $μα[β]$, so this action becomes
inactive and it is bisimilar to the empty process.
As above, this reduction is the only one possible.

\subsection{Simulation in call-by-value} 

The substitution rule for continuations and the equivalence rule for $μα[β]$
hold by the same arguments as in the case of call-by-name.

For the first equivalence, we have
\begin{align*}
  ⟦ K * (M)N ⟧
  &= \new{u}
    \bigl( ⟦u=K⟧
    \bigm| \new{v}( ⟦v=uM^f⟧ | ⟦N⟧v )
    \bigr) \\
  &≡ \new{v}
    \bigl( \new{u}( ⟦u=K⟧ | ⟦v=uM^f⟧ )
    \bigm| ⟦N⟧v
    \bigr) \\
  &= ⟦ KM^f * N ⟧
\end{align*}
For the first reduction rule, we have
\begin{align*}
  ⟦ KM^f * V ⟧
  &= \new{u}
    \Bigl( \new{v}
      \bigl( \bar{γ}'u(x).\new{w}( ⟦M⟧w | \bar{δ}w⟨xv⟩ )
      \bigm| ⟦v=K⟧ \bigr)
    \Bigm| γ'u(x).⟦x=V⟧
    \Bigr) \\
  &≡ \new{uv}
    \bigl( \bar{γ}'u(x).\new{w}( ⟦M⟧w | \bar{δ}w⟨xv⟩ )
    \bigm| ⟦v=K⟧
    \bigm| γ'u(x).⟦x=V⟧
    \bigr) \\
  &→ \new{uvx}
    \bigl( \bar{γ}'u(x).\new{w}( ⟦M⟧w | \bar{δ}w⟨xv⟩ ) \\&\qquad
    \bigm| \new{w}( ⟦M⟧w | \bar{δ}w⟨xv⟩ )
    \bigm| ⟦v=K⟧
    \bigm| ⟦x=V⟧
    \bigr) \\
  &≡ \new{w}
    \bigl( \new{u} \bar{γ}'u(x).\new{w}( ⟦M⟧w | \bar{δ}w⟨xv⟩ ) \\&\qquad
    \bigm| ⟦M⟧w
    \bigm| \new{vx}( ⟦x=V⟧ | \bar{δ}w⟨xv⟩ | ⟦v=K⟧ )
    \bigr) \\
  &≅ ⟦ KV^a * M ⟧
\end{align*}
where $→$ contains one transition for each modality in the word $γ'$.
In the classical case $\bar{γ}'$ contains $!$ so the continuation at $u$ is
not consumed, however we know that $u$ has no other occurrence since
continuations are duplicated by the rule for $μ$, so we can erase the
residual term on $u$ by bisimilarity.
This is the only possible reduction as soon as $γ'$ is not empty.
The second reduction rule is deduced as
\begin{align*}
  ⟦ KW^a * V ⟧
  &= \new{u}
    \bigl( \new{vx}( ⟦x=W⟧ | \bar{δ}u⟨xv⟩ | ⟦v=K⟧ )
    \bigm| γ'u(z).⟦z=V⟧
    \bigr) \\
  &≡ \new{uvx}
    \bigl( ⟦x=W⟧
    \bigm| \bar{δ}u⟨xv⟩
    \bigm| ⟦v=K⟧
    \bigm| γ'u(z).⟦z=V⟧
    \bigr) \\
  &→ \new{uvx}
    \bigl( ⟦x=W⟧
    \bigm| \bar{γ}z⟨xv⟩
    \bigm| ⟦v=K⟧
    \bigm| ⟦z=V⟧
    \bigr) \\
  &≡ ⟦ K * V⋅W ⟧
\end{align*}
where $→$ contains one transition for each modality in the word $γ'$, since
$δ=γ'γ$.
As above, this is the only reduction.
For the substitution rule, we have
\begin{align*}
  ⟦ K * λx.M⋅V ⟧
  &= \new{u}
    \bigl( ⟦u=K⟧
    \bigm| \new{vz}( γv(xw).⟦M⟧w | ⟦z=V⟧ | \bar{γ}v⟨zu⟩ )
    \bigr) \\
  &→ \new{u}
    \bigl( ⟦u=K⟧
    \bigm| \new{x}( ⟦M⟧u | ⟦x=V⟧ )
    \bigr) \\
  &≅ \new{u}
    \bigl( ⟦u=K⟧
    \bigm| ⟦M[V/x]⟧u
    \bigr) \\
  &= ⟦ K * M[V/x] ⟧
\end{align*}
where there is one transition for each modality in $γ$.
The step after the reduction is an instance of the substitution lemma
$\new{x}(⟦M⟧u|⟦x=V⟧)≅⟦M[V/x]⟧$.
This lemma holds by the same argument as in the case of call-by-name: the
binding $⟦x=V⟧$ can be distributed to all occurrences of $x$, but any occurrence
of $x$ occurs in a binding $⟦y=x⟧$ so we have
\begin{multline*}
  \new{x}( ⟦y=x⟧ | ⟦x=M⟧ )
  = \new{x}( y\fuse x | ⟦x=M⟧ ) \\
  ≡ \new{x}(y\fuse x) | ⟦y=M⟧
  ≅ ⟦y=M⟧
\end{multline*}
using the obvious bisimilarity $\new{x}(x\fuse y)≅1$.

\end{document}